\documentclass[11pt,letterpaper]{article}

\usepackage[margin=1in]{geometry}
\usepackage{amsmath,amssymb,amsthm,mathtools}
\usepackage{graphicx}
\usepackage{tabularx,booktabs}
\usepackage{microtype}
\usepackage{float}
\usepackage{enumitem}
\usepackage{xcolor}
\usepackage[colorlinks=true,linkcolor=blue!45!black,citecolor=blue!45!black,urlcolor=blue!45!black]{hyperref}
\usepackage[capitalise,nameinlink,noabbrev]{cleveref}
\setlist[enumerate]{leftmargin=2em,topsep=3pt,itemsep=3pt,parsep=1pt}

\newcommand{\N}{\mathbb{N}}
\newcommand{\Z}{\mathbb{Z}}
\newcommand{\bits}{\{0,1\}}
\newcommand{\EXP}{\mathsf{EXP}}
\newcommand{\NEXP}{\mathsf{NEXP}}
\newcommand{\Square}{\mathsf{Square}}
\newcommand{\Rect}{\mathsf{Rect}}
\newcommand{\Cfg}{\mathsf{Cfg}}
\newcommand{\plain}[1]{\mathsf{plain}(#1)}
\newcommand{\hd}[2]{\mathsf{hd}(#1,#2)}

\newcommand{\Junk}{\mathsf{Junk}}
\newcommand{\Part}{\mathsf{Part}}
\newcommand{\Car}{\mathsf{Car}}
\newcommand{\Cav}{\mathsf{Cav}}
\newcommand{\Sites}{\mathsf{Sites}}

\newcommand{\NewMid}{\mathsf{NewMid}}
\newcommand{\ModeOut}{\mathsf{ModeOut}}
\newcommand{\Emit}{\mathsf{Emit}}
\newcommand{\LegalScan}{\mathsf{LegalScan}}

\newcommand{\binenc}{\mathsf{bin}}
\newcommand{\val}{\mathsf{val}_2}
\newcommand{\supp}{\operatorname{supp}}

\newcommand{\Tstar}{\mathcal{T}^{\star}}
\newcommand{\Lplug}{L_{\mathrm{plug}}}

\theoremstyle{plain}
\newtheorem{thm}{Theorem}[section]
\newtheorem{lem}[thm]{Lemma}
\newtheorem{prop}[thm]{Proposition}

\theoremstyle{definition}

\theoremstyle{remark}

\title{Counting Square Tilings with a Fixed Set of Polyominoes Is \texorpdfstring{\(\#\mathrm{EXP}\)}{\#EXP}-Complete}
\author{%
  Qiuyu Chen$^{1,2}$ \qquad Baao Xie$^{2,*}$\\[6pt]
  {\normalsize $^{1}$Shanghai Jiao Tong University}\\[3pt]
  {\normalsize $^{2}$Ningbo Key Laboratory of Spatial Intelligence and Digital Derivative,}\\
  {\normalsize Ningbo Institute of Digital Twin, Eastern Institute of Technology, Zhejiang, China}
}
\date{}

\begin{document}
\maketitle
\begingroup
\renewcommand{\thefootnote}{*}%
\footnotetext{Corresponding author.}%
\endgroup
\begin{abstract}
An important question in counting complexity is whether exact counting can remain hard under restrictions that make feasibility easy. Tilings of squares by a fixed family of polyominoes provide a stringent test of this question, since tileability is polynomial-time decidable and the only varying input is a binary-encoded side length. We prove that there exists a fixed finite family of connected polyominoes for which counting translation-only square tilings is $\#\mathrm{EXP}$-complete under polynomial-time metric reductions, while the counting function belongs to $\#\mathrm{E}$. The resulting separation reveals that solution multiplicities can retain computational information that feasibility alone cannot capture. The proof develops a two-stage numerical isolation technique that combines a boundary-amplified Wang system with a rigid geometric compiler to recover accepting-computation counts while controlling all noncanonical tilings.
\end{abstract}

\section{Introduction}\label{sec:introduction}

Feasibility and exact counting can have sharply different computational complexity. This difference is particularly revealing for problems whose inputs leave little freedom to encode computation. Polyomino tilings provide a classical geometric setting for investigating this issue, since the choice of permitted pieces and the shape of the region offer two distinct sources of variation. Restricting these sources in turn brings into focus how much computational information can remain in the number of tilings even when feasibility is easy to determine.

Classical results illustrate how the representation of a tiling instance affects its complexity. Domino tilings admit efficient counting algorithms based on Pfaffian methods \cite{Kasteleyn61}, while more general tiling problems exhibit substantial computational hardness \cite{MooreRobson01}. Pak and Yang \cite{PakYangRect13} proved that counting tilings of simply connected regions by a fixed finite family of rectangles is $\#\mathrm{P}$-complete. Thus, fixing the pieces still leaves the boundary of the input region available to encode computation. A stronger restriction is to require every region to be an ordinary square, leaving its side length as the only varying parameter.

For a fixed finite family $T$ of polyominoes, let $g_T(n)$ be the number of tilings of the $n\times n$ square using translations of pieces in $T$. With $n$ given in binary, $O(\log n)$ bits can describe a square containing $n^2$ cells. The central question is whether this restricted input can still yield a $\#\mathrm{EXP}$-complete counting function. This fixed-family square-counting problem was formulated by Pak \cite[Conjecture~1.7]{Pak18}. Moreover, the tileability of an $m\times n$ rectangle can be decided in $O(\log m+\log n)$ time for every fixed family of tiles \cite[Theorem~1.9]{Pak18}. The setting therefore offers a stringent test of whether solution counts can retain computational information beyond what feasibility reveals.

Our main theorem answers this question affirmatively. Here $\#\mathrm{EXP}$ counts accepting paths of nondeterministic machines running in time $2^{\operatorname{poly}(\ell)}$ on inputs of length $\ell$, and $\#\mathrm{E}$ denotes the subclass with time bound $2^{O(\ell)}$. Hardness is understood under polynomial-time metric reductions \cite{Kren88}, which recover the source count by polynomial-time postprocessing of a target count given in binary.

\begin{thm}[Main theorem]\label{thm:main}
There exists a fixed finite family $\Tstar$ of nonempty connected polyominoes such that counting translation-only tilings of the $n\times n$ square, with $n$ given in binary, is $\#\mathrm{EXP}$-complete under polynomial-time metric reductions. The counting function belongs to $\#\mathrm{E}$.
\end{thm}

The $\#\mathrm{E}$ upper bound follows by canonically encoding a tiling and checking it in time exponential in the bit length of $n$. The central challenge is hardness. The count for an unmodified square includes every tiling admitted by its fixed pieces, not just configurations representing an encoded computation. Recovering the desired accepting-path count therefore requires control over the aggregate contribution of all these configurations. The proof achieves this by numerical isolation, arranging for the intended computation histories to contribute a common multiplicative factor while the entire remaining contribution is smaller than that factor. Integer division then recovers the target count exactly. This principle is implemented at two successive levels.

The first level uses a fixed Wang-tile system to encode accepting computations. Local edge-matching rules have long been used to simulate computation \cite{Wang61,Berger66,Robinson71}, and fixed translationally invariant systems with binary-encoded dimensions provide related examples of computational hardness \cite{GottesmanIrani13,AharonovIrani23}. A fixed-Wang-square counting result was also reported by Mennen and Pak \cite{PakTalk20}. In the present construction, a scanner reconstructs a machine's initial configuration from the binary side length, and a Wang tableau records its accepting branches. A boundary amplifier gives the intended full-frame histories a dominant multiplicity, while a bound on all other boundary modes allows their contribution to be separated arithmetically. Crucially, the same construction bounds the number of open Wang tilings of any rectangle exponentially in its perimeter, providing the control needed by the geometric stage.

The second level realizes the Wang system using a fixed family of connected polyominoes. Geometric encodings of local matching rules provide a starting point \cite{DemaineDemaine07,Ollinger09}; here the central issue is preserving exact counts over all geometric tilings. Each Wang tile is represented by a carrier whose interior cavity records its identity and whose interfaces encode matching data. Local rigidity aligns the carriers on a common macrogrid, and the open-interface structure partitions the grid into rectangles. This gives an exact decomposition of the polyomino tiling count into contributions from rectangular dissections. The single-block configuration carries the Wang count, whereas every nontrivial dissection loses interface-label multiplicity along its seams. Bounds on seam length, the number of dissections, and the Wang counts within their blocks together make the total noncanonical contribution strictly smaller than one unit after normalization. A second integer division therefore extracts the Wang count.

Composing the two numerical readouts gives the required polynomial-time metric reduction. More broadly, numerical isolation combines a dominant term encoding accepting computations with global combinatorial bounds on the remaining configurations. Together with rigid geometric compilation, this provides a route to exact counting hardness in other succinctly specified systems with fixed local rules.

\section{Preliminaries}\label{sec:preliminaries}

The tiling conventions, counting complexity classes, metric reductions, and machine notation are specified below.

All regions are subsets of \(\Z^2\); a finite set is 4-connected if its grid-adjacency graph is connected. Let \(\binenc(n)\) denote the canonical binary representation of \(n\), with \(\binenc(0)=0\). For \(n\geq0\), write \(\Square_n=\{0,\ldots,n-1\}^2\), with \(\Square_0=\varnothing\). Rectangles of positive width \(w\) and height \(h\) are denoted by \(\Rect_{w,h}=\{0,\ldots,w-1\}\times\{0,\ldots,h-1\}\).

If \(T=\{\tau_1,\ldots,\tau_k\}\), a \(T\)-tiling of a finite region \(\Omega\) is a finite set of placements
\((t,v)\in\{1,\ldots,k\}\times\Z^2\) whose translates \(\tau_t+v\) are pairwise disjoint and have union
\(\Omega\).  Only translations are allowed.  The empty region has one empty
tiling.

A Wang system is a finite set of labelled tile types, each carrying north, east, south, and west colours.  Distinct tile types are allowed to carry the same four colours.  An open Wang tiling of a finite rectangle assigns a tile type to each cell, matches colours on internal edges, and imposes no condition on colours exposed along the geometric boundary.  This type-level convention is important in Section~\ref{sec:geometry}, where the geometric compiler preserves tile identity even when four edge colours coincide.

The class \(\#\EXP\) is used in its standard complexity-theoretic sense: it consists of total functions \(f\colon\bits^*\to\N\) equal to the number of accepting paths of a nondeterministic Turing machine that halts within
\(2^{p(|x|)}\) steps for some polynomial \(p\).  Likewise, \(\NEXP\) denotes nondeterministic time \(2^{\operatorname{poly}(|x|)}\).  The linear-exponential counting subclass is denoted by \(\#\mathrm{E}:=\#\mathrm{NTIME}(2^{O(|x|)})\).  All counting functions in this paper are total: malformed or noncanonical binary strings receive value zero.

\begin{prop}[Fixed-family square counts lie in \(\#\mathrm{E}\)]\label{prop:fixed-family-sharpE}
For every fixed finite family \(T\) of finite polyominoes, the function \(G_T\) belongs to \(\#\mathrm{E}\).
\end{prop}

\begin{proof}
Let \(\ell=|\binenc(n)|\).  A tiling of \(\Square_n\) has at most \(n^2\) placements.  Encode the placements in a canonical lexicographic order, recording for each one its fixed tile type and its anchor coordinates.  Because \(T\) is fixed, the resulting witness has \(O(n^2\ell)\) bits, which is \(2^{O(\ell)}\).  A deterministic verifier checks containment, pairwise consistency, and exact coverage in time polynomial in this witness length.  The canonical ordering gives one accepting witness per tiling.  Since \(n<2^\ell\), both witness generation and verification take \(2^{O(\ell)}\) time.
\end{proof}

A polynomial-time metric reduction from \(f\colon\bits^*\to\N\) to \(g\colon\bits^*\to\N\) is a pair of polynomial-time functions \(r\) and \(h\) satisfying \(f(x)=h\bigl(x,g(r(x))\bigr)\) for every \(x\), where the integer argument of \(h\) is given in binary.  Thus the running time of \(h(x,y)\) is measured in \(|x|+|\binenc(y)|\).  This is Krentel's metric-reduction model \cite{Kren88}. For a counting function \(f\), write \(\supp(f)=\{x:f(x)>0\}\) for its support language.

For a one-tape machine with tape alphabet \(\Gamma\), state set \(Q\), and sentinel symbol \(\mathsf{\#_s}\), a cell is either \(\plain{a}\) or \(\hd{q}{a}\).  A \emph{legal extended configuration of height \(H\geq3\)} is a word on indices \(-1,0,\ldots,H\) with sentinel cells \(\plain{\mathsf{\#_s}}\) at \(-1,0,H-1,H\), exactly one head cell at an index in \(\{1,\ldots,H-2\}\), and plain cells elsewhere.  A labelled arrow is a tuple \(\alpha=(q,a,q',b,D)\), with \(D\in\{L,S,R\}\); an arrow is inapplicable if its requested move would enter a sentinel.  Distinct labelled arrows count as distinct computation branches even when they induce the same tape update.

\subsection{Proof strategy}\label{subsec:counting-argument}

The hardness proof proceeds through two numerical readouts. The first uses a fixed nondeterministic machine whose accepting-path count $C_U(x)$ is $\#\mathrm{EXP}$-complete. A polynomial-time map encodes $x$ as a positive binary side length $n=E_{\mathrm{side}}(x)$. The Wang construction of Sections~\ref{sec:tableau}--\ref{sec:wang} has a total open-square tiling count $a_{\mathrm{prod}}(n)$. For all sufficiently large encoded side lengths, its full boundary class contributes a fixed multiple of $1+C_U(x)$, while the other boundary classes contribute a remainder $J_W(n)$ satisfying
\begin{equation}\label{eq:roadmap-wang}
 a_{\mathrm{prod}}(n)=Q^{4n-4}(1+C_U(x))+J_W(n),
 \qquad 0\le J_W(n)<Q^{4n-4}.
\end{equation}
The full boundary frame has $4n-4$ independently labelled sites. Every proper boundary mode has fewer such sites, and the number of compatible tableaux can be controlled by a perimeter-scale bound. Division by $Q^{4n-4}$ therefore retrieves the accepting-path count, with the finitely many exceptional inputs handled by a fixed table.

The second readout is geometric. Section~\ref{sec:geometry} represents the fixed Wang system by connected polyominoes with a constant macrocell pitch $P_s$ and $s_{\mathrm{amp}}$ possible labels on each open interface. The total number $E_n=2n(n-1)$ of macrogrid adjacencies gives the canonical one-block contribution $s_{\mathrm{amp}}^{E_n}a_{\mathrm{prod}}(n)$. All other geometric tilings are classified by rectangular partitions, whose seams reduce the number of available labels. Denoting their aggregate contribution by $J_P(n)$, the geometric counting identity yields
\begin{equation}\label{eq:roadmap-geo}
 g_{\Tstar}(P_sn)=s_{\mathrm{amp}}^{E_n}a_{\mathrm{prod}}(n)+J_P(n),
 \qquad 0\le J_P(n)<s_{\mathrm{amp}}^{E_n}.
\end{equation}
This second inequality follows from a lower bound on seam length, an exponential bound on the number of rectangular partitions, and the perimeter bound for open Wang rectangles. The exact geometric interfaces are defined in Section~\ref{sec:geometry}, where Figures~\ref{fig:carrier-interface} and~\ref{fig:dissection-suppression} accompany their respective counting arguments.

Together, these identities describe how the accepting-path count is encoded first in the open Wang-square count and then in the translation-only polyomino-square count. The first numerical readout separates the intended computation histories from other Wang boundary modes, while the second separates the canonical one-block contribution from all other polyomino tilings. Applied in reverse order, the two floor divisions recover the original accepting-path count. Sections~\ref{sec:first-readout} and~\ref{sec:completion} establish the bit-complexity bounds required for the resulting metric reductions.

\section{Encoding an Exponential-Time Counting Computation}\label{sec:source}

The hardness reduction begins with a fixed counting machine whose initial configuration can be reconstructed from the square's side length. Parsimonious simulation and polynomial padding provide a universal accepting-path counter. A self-delimiting encoding then places its input in the binary side length; the final bounds ensure that the encoded computation can be recovered and completed within the available tableau.

\begin{thm}[Computation encoded by the side length]\label{thm:source-machine}
There are a fixed one-head nondeterministic Turing machine $M_\star$, a counting function $C_U\in\#\mathrm E$ that is $\#\mathrm{EXP}$-complete under parsimonious reductions, and a polynomial-time injective map $E_{\mathrm{side}}:\bits^*\to\N$ with the following property. Let $\operatorname{Acc}_\star(n)$ be the number of accepting length-$n$ branches of $M_\star$ from the initial configuration determined by the binary digits of $n$. Then, for every $n\ge 5$,
\[
\operatorname{Acc}_\star(n)=
\begin{cases}
1+C_U(x),& n=E_{\mathrm{side}}(x),\\
1,&n\notin E_{\mathrm{side}}(\bits^*).
\end{cases}
\]
On encoded side lengths $n=E_{\mathrm{side}}(x)$, the parser, input preparation, and simulation reach an accepting or rejecting stay-sink within $n-1$ steps. Both stay-sinks have unique continuations.
\end{thm}

\begin{proof}
We first normalize an arbitrary \(\#\EXP\) counting machine without changing its path count.  Let the source machine be fixed and suppose that, on inputs of length \(s\), every branch halts within \(T(s)=2^{p(s)}\) steps for a polynomial \(p\).  Enlarging \(p\) if necessary, assume \(p(s)\geq s+1\).  A finite transition list is encoded self-delimitingly by unary indices for states and letters and fixed unary codes for the three head movements.  The decoding procedure is deterministic and polynomial in the code length.

For later use, suppose a source pair has at most \(d_{\max}\) outgoing arrows.  Its choice can be replaced by a complete binary tree of depth
\[
 k_{\mathrm{tree}}=
 \begin{cases}
 1,&d_{\max}\leq1,\\
 \min\{k:2^k\geq d_{\max}\},&d_{\max}\geq2.
 \end{cases}
\]
Leaves representing unused choices enter a rejecting sink.  Accepting and rejecting sinks each have one stay-arrow on every tape letter.  A selected original arrow determines one tree path, so the conversion is bijective on labelled computation paths.

For a machine with \(t\) tapes, store all tapes in one contiguous block, recording a letter and a head bit for each simulated tape.  After \(u\) steps the block length is at most \(t(s+u+1)\).  A scan, rewrite, and head-bit update costs at most
\(2t(s+u+1)+16\).  Summation for \(0\leq u<T(s)\) gives
\[
 T(s)\bigl(2t(s+T(s)+1)+16\bigr)
 \leq 2^{2p(s)+c_t}
\]
for a constant \(c_t\) depending only on the fixed source machine, after enlarging it to cover finitely many small inputs.  Applying the binary-tree replacement above to each nondeterministic choice has constant depth, so the fully converted one-tape binary machine halts within \(2^{q(s)}\) steps for some polynomial \(q\) depending only on the source machine.  All work inserted between nondeterministic choices is deterministic, so the path count is unchanged.

The fixed machine \(U_{\log}\) parses strings of the form
\[
 1^{k}0\,\mathsf{code}_M\,1^{p}0z,
 \qquad |\mathsf{code}_M|=k,
\]
where the unary block \(1^p\) is padding.
It converts the decoded transition list once and simulates the resulting binary machine on \(z\) for exactly \(2^m\) simulated steps, where \(m\) is the total input length.  Missing transitions have already been redirected to the rejecting sink, and after the simulated-step budget any branch that has not reached a stay-sink is sent deterministically to the rejecting stay-sink.  A simulated tape block has length at most \(2^m+2m+4\), so one simulated step costs at most \(2^{2m+2}\) for \(m\geq1\).  Including parsing, the running time is at most \(m^2+4m+32+2^{3m+2}+2^m\leq 2^{8(m+1)}\). The inequality is checked directly for \(m=0,1\); for \(m\geq2\),
\(m^2+4m+32\leq2^{3m}\).  Set \(p_U(m)=8(m+1)\). Thus \(U_{\log}\) reaches a stay-sink within \(2^{p_U(m)}\) steps.

Let \(C_U(x)\) count the paths in the accepting sink at that time.  A binary clock of \(p_U(m)+1\) bits makes a genuinely halting simulator.  Its tape length is at most \(m+2^{p_U(m)}+4\); hence its running time is at most
\[
 2^{p_U(m)}
 \bigl(m+2^{p_U(m)}+p_U(m)+8\bigr)^2
 \leq2^{3p_U(m)+2}
 \leq2^{26(m+1)}.
\]
For hardness, let \(f\in\#\EXP\) and apply the preceding conversions to a fixed machine computing \(f\).  Let \(\mathsf{code}_f\) be its transition code before the binary-tree conversion, and let \(q_f(s)\) be a polynomial such that the converted machine halts within \(2^{q_f(s)}\) steps on every input of length \(s\).  Choose a fixed constant \(c_f\) large enough to cover the delimiters and code, and define the parsimonious source map
\[
 r_f(z)=
 1^{|\mathsf{code}_f|}0\,\mathsf{code}_f\,
 1^{q_f(|z|)+c_f}0z.
\]
Its output length \(m\) is polynomial in \(|z|\) and satisfies \(m\geq q_f(|z|)\).  Hence the \(2^m\)-step universal simulation outlasts every converted computation.  Once the simulated machine reaches an accepting or rejecting stay-sink, the continuation is unique, so accepting paths are preserved bijectively.  Thus
\(
 C_U(r_f(z))=f(z)
\)
for every \(z\), and \(C_U\) is \(\#\EXP\)-complete under parsimonious polynomial-time reductions.  Notice that \(C_U\) itself is computable by a nondeterministic machine in linear-exponential time; the polynomial padding in \(r_f\) is what allows hardness for the full standard \(\#\EXP\) class.

Moreover, if \(L\in\NEXP\), the function counting accepting paths of a nondeterministic \(2^{\operatorname{poly}(|z|)}\)-time machine deciding \(L\) belongs to \(\#\EXP\) and has support exactly \(L\).  Parsimonious completeness therefore implies that the support language \(\supp(C_U)\) is \(\NEXP\)-hard.

For \(m=|x|\), define the self-delimiting word and its numerical value by
\[
 \beta(x)=1x1\,0^m1\,0^{72m+72},
 \qquad E_{\mathrm{side}}(x)=\val(\beta(x)).
\]
Its length is \(L_{\mathrm{enc}}(x)=74m+75=8p_U(m)+10m+11\). Reading from the least significant end, count the terminal zero run, cross its following separator \(1\), and count the next zero run of length \(m\). The next separator \(1\) is then checked and skipped before reading the preceding \(m\) bits as \(x\); the remaining leading bit must be \(1\), with no further bits.  Checking that the first run has length \(72m+72\) recovers \(x\) uniquely.  The procedure scans only polynomially many bits and proves injectivity and polynomial-time decodability.

Since \(E_{\mathrm{side}}(x)\geq2^{L_{\mathrm{enc}}(x)-1}\),
\[
 \frac{L_{\mathrm{enc}}(x)-1}{4}
 =2p_U(m)+\frac{5m}{2}+\frac52
 \geq2p_U(m)+2.
\]
Writing \(t=2^{p_U(m)}\), we have \(4t^2\geq t+1\), and therefore
\[
 2^{p_U(m)}+1
 \leq2^{(L_{\mathrm{enc}}(x)-1)/4}
 \leq E_{\mathrm{side}}(x)^{1/4}.
\]
Also \(L_{\mathrm{enc}}(x)\geq75\), giving
\(E_{\mathrm{side}}(x)\geq2^{74}\).  Binary branching further gives
\(1+C_U(x)\leq2^{E_{\mathrm{side}}(x)^{1/4}}\).

We next define \(M_\star\).  Its alphabet contains that of \(U_{\log}\), the sentinel \(\mathsf{\#_s}\), delimiters \(\mathsf{markL},\mathsf{markR}\), and finitely many parser letters.  At its start pair it has exactly two arrows.  The null arrow enters the accepting stay-sink.  The real arrow enters a deterministic parser.

Write \(\Cfg^\star_n\) for the initial configuration specified by the binary digits of \(n\). On \(\Cfg^\star_n\), the sentinels occupy \(-1,0,n-1,n\).  The head is at index \(1\), on the least significant bit.  For \(1\leq y\leq n-3\), index \(y+1\) stores \(\mathsf{one_b}\) if the \(y\)-th bit of \(n\) is one, \(\mathsf{zero_b}\) if it is zero below the leading bit, and \(\mathsf{work_b}\) above the leading bit.  The inequalities \(n<2^{n-2}\) (for \(n\geq5\)) and \(\lfloor\log_2 n\rfloor\leq n-3\) show that the two omitted high bits are zero and that the leading bit remains inside the sentinel interval.  At \(n=4\), the omitted bit at position \(2\) is one.

The parser reads terminal zeros, the two separators, and the candidate \(x\), and checks \(R_{\mathrm{pad}}(m)=72m+72\).  It then inspects the cell immediately north of the recovered leading one.  A \(\mathsf{work_b}\) accepts the high-zero continuation; a bit letter or sentinel rejects.  Thus a successful parse occurs exactly when \(n=E_{\mathrm{side}}(x)\).

After a successful parse, the machine writes \(\mathsf{markL}\), then
\(2^{p_U(m)}+1\) blanks, a copy of \(x\), another
\(2^{p_U(m)}+1\) blanks, and \(\mathsf{markR}\).  It restores parser scratch and starts \(U_{\log}\) on the copied bit letters.  The head cannot reach either mark during \(2^{p_U(m)}\) steps.  The copy remains between the sentinels because
\[
 2^{p_U(m)+1}+L_{\mathrm{enc}}(x)+m+6
 \leq2^{p_U(m)+3}
 \leq8E_{\mathrm{side}}(x)^{1/4}
 <E_{\mathrm{side}}(x)-2.
\]

The parser takes at most \((L_{\mathrm{enc}}(x)+4)^4\) steps, the copy at most
\(m2^{2p_U(m)+6}\), and the simulation at most \(2^{p_U(m)}\).  For
\(n=E_{\mathrm{side}}(x)\geq2^{74}\), the first term is at most \(n/4\), the second at most \(64n^{5/8}\leq n/4\), and the third at most \(n^{1/4}\leq n/4\).  Hence the real computation enters a sink by time \(n-1\).  No transition enters a sentinel, and the remaining steps are uniquely supplied by the appropriate stay-sink.

The null branch contributes one accepting path for every \(n\geq5\).  On a coded side, real accepting branches are in bijection with accepting paths of \(U_{\log}\); off-code, the parser enters the rejecting sink.  This proves the branch-count formula and all remaining assertions.
\end{proof}

\section{Counting Paths with Wang Tile Tableaux}\label{sec:tableau}

Each step of the size-encoded machine is represented by one column of Wang tiles. The local rules record the unique head and its labelled transition, yielding a bijection between accepting branches and tableaux with prescribed boundaries. A separate count of successor columns under arbitrary boundary data provides the perimeter estimate needed in the geometric construction.

Let \(M\) be a finite one-tape machine with state set \(Q_{\mathrm{st}}\), tape alphabet \(\Gamma_{\mathrm{tp}}\), accepting state \(q_{\mathrm{acc}}\), sentinel \(\mathsf{\#_s}\), and labelled-arrow set \(\Delta_{\mathrm{arr}}\).  Put
\[
 \Sigma_{\mathrm{cell}}
 =\{\plain{a}:a\in\Gamma_{\mathrm{tp}}\}
 \sqcup
 \{\hd{q}{a}:q\in Q_{\mathrm{st}},\,a\in\Gamma_{\mathrm{tp}}\},
 \qquad
 s_{\mathrm{cell}}
 =|\Gamma_{\mathrm{tp}}|(1+|Q_{\mathrm{st}}|).
\]

\begin{thm}[Wang tableau for counting computation paths]\label{thm:tableau}
For every fixed one-head Turing machine $M$, there exists a finite Wang system $W_{\mathrm{tab}}$ with the following properties. For any legal initial configuration of height $H$ and any time horizon $T$, accepting labelled-arrow branches of length $T$ correspond bijectively to $W_{\mathrm{tab}}$-tilings of a $T$-by-$H$ rectangle with prescribed initial, accepting, and sentinel boundary data. Distinct labelled arrows yield distinct tableaux.

Furthermore, writing $s_{\mathrm{cell}}$ for the number of possible tape-cell symbols and $d_{\mathrm{arr}}$ for the number of labelled arrows, a prescribed east-boundary colour word of height $n$ has at most \(K_{\mathrm{col}}=4s_{\mathrm{cell}}^4(d_{\mathrm{arr}}+1)^4\) vertically matching successor columns. The number of open $n$-square tableaux therefore satisfies
\[
a_{\mathrm{tab}}(n)\le R_{\mathrm{tab}}^n,
\qquad R_{\mathrm{tab}}=4s_{\mathrm{cell}}^7(d_{\mathrm{arr}}+1)^4.
\]
\end{thm}

\begin{proof}
For an old radius-one word \(X=(x_-,x_0,x_+)\), define its local update under an applicable arrow \(\alpha=(q,a,q',b,D)\) when \(x_0=\hd{q}{a}\), the two neighbours are tape cells, and a left or right movement does not enter a sentinel.  The updated triple is
\[
 \Theta_\alpha(X)=
 \begin{cases}
  (\hd{q'}{\ell},\plain{b},\plain{r}),&D=L,\\
  (\plain{\ell},\hd{q'}{b},\plain{r}),&D=S,\\
  (\plain{\ell},\plain{b},\hd{q'}{r}),&D=R,
 \end{cases}
\]
where \(x_-=\plain{\ell}\) and \(x_+=\plain{r}\).  For a window whose head is immediately south or north of its centre, let
\(\rho_\alpha\) and \(\lambda_\alpha\) perform the corresponding incoming-head update at the centre.

Horizontal colours are symbols \(\mathsf{hcol}(X)\) for
\(X\in\Sigma_{\mathrm{cell}}^3\).  A candidate tile carries two order-bit chains, one for the old word and one for the new word.  Write their south/north values as \((o_S,o_N)\) and \((n_S,n_N)\), both in \(\bits^2\).
For the old chain we require \(o_N=o_S\) when the old centre is not a head and \((o_S,o_N)=(0,1)\) when it is a head; the new chain satisfies the analogous rule with the new centre.  Thus each chain is nondecreasing and can jump only once, at its unique head.  Vertical colours have the form \(\mathsf{vcol}(P_{\mathrm{old}},P_{\mathrm{new}},o,n,\mu)\),
where both \(P\)'s are pairs of cell symbols, \(o,n\in\bits\) are the old- and new-chain values on that horizontal boundary, and
\(\mu\in\Delta_{\mathrm{arr}}\sqcup\{\bot\}\).  The south colour stores \((o_S,n_S)\), and the north colour stores \((o_N,n_N)\).

A candidate tile contains old and new triples \(X,Y\), these four order bits, and south and north tokens.  It is admitted in exactly one of four cases.
\begin{enumerate}
\item[\textup{H}.] The old centre is a head.  The south token is \(\bot\), the north token is an applicable arrow \(\alpha\), and
\(Y=\Theta_\alpha(X)\).
\item[\textup{A}.] The old south neighbour is the head.  The south token is
\(\alpha\), the north token is \(\bot\), and the new centre is
\(\rho_\alpha(x_0)\).
\item[\textup{B}.] The old north neighbour is the head.  Both tokens are
\(\bot\), and the new centre equals \(\lambda_\alpha(x_0)\) for at least one arrow \(\alpha\) having the source displayed by that head.  This is an existence condition on the tile datum: if several labelled arrows witness the same resulting datum \((X,Y,o_S,o_N,n_S,n_N)\), they do \emph{not} create duplicate tile types in Case~B.
\item[\textup{Q}.] The old triple has no head.  Both tokens are \(\bot\), and the new centre equals the old centre.
\end{enumerate}
Both triples are required to contain at most one head.  The east and west colours are \(\mathsf{hcol}(X)\) and \(\mathsf{hcol}(Y)\).  The south vertical colour stores the two southern old and new symbols, the lower order bits, and the south token; the north colour stores the two northern symbols, the upper bits, and the north token.  The inventory is finite and has the asserted cardinality.  For the conditioned tableau used in the bijection, choose the south sentinel colour with both reconstructed symbol pairs equal to sentinels, both order-chain values equal to zero, and empty token; choose the north sentinel colour analogously with both order-chain values equal to one and empty token.

Vertical matching overlaps consecutive symbol pairs and identifies the corresponding boundary bits.  Hence a matching column reconstructs unique extended old and new words on indices \(-1,\ldots,H\), together with one old order chain and one new order chain.  Each chain is constant except for a possible \(0\)-to-\(1\) jump at a head centre, so it certifies at most one head.  When its prescribed south boundary value is zero and north boundary value is one, it certifies exactly one head.  Sentinel vertical colours additionally force the two bottom and two top symbols to be \(\plain{\mathsf{\#_s}}\), so every certified head lies strictly between them.

At the unique old head, Case H chooses one arrow and places its label on the north edge.  Vertical matching passes that token to the Case A tile.  The overlap with \(Y=\Theta_\alpha(X)\) simultaneously fixes the Case B centre for the same arrow.  Every window at distance at least two is Case Q and preserves its centre.  Thus the old word and the one arrow determine the entire visible new word.  This proves that the neighbouring windows make no independent choice.  Because the token is the labelled arrow itself, two same-effect arrows give different vertical colours.

Time runs from east to west.  Fixing the east boundary to the radius-one windows of \(\Cfg_0\) determines the first old word.  Horizontal matching identifies each column's new word with the next column's old word.  The sentinel colours make every word a legal configuration, and the accepting west-colour set requires the final unique head to have state \(q_{\mathrm{acc}}\).  Conversely, a branch determines all triples, bits, and tokens and hence one tiling.  These constructions are inverse.

For the open-column bound, fix its old-word triples.  A column is encoded injectively by the two south-boundary order bits (one for each chain), the four new endpoint and ghost symbols, and four token slots containing the exterior tokens and the possible visible-head token.  The old word together with its chosen south-boundary bit fixes the old order chain.  Cases H--Q, the new south-boundary bit, and the recorded endpoint data recover the new word, its order chain, and every token.  Therefore the number of columns is at most \(2^2s_{\mathrm{cell}}^4(d_{\mathrm{arr}}+1)^4=K_{\mathrm{col}}\). The east boundary of an open \(n\)-square has at most
\((s_{\mathrm{cell}}^3)^n\) words.  Scanning \(n\) columns westward gives
\[
 a_{\mathrm{tab}}(n)
 \leq(s_{\mathrm{cell}}^3)^nK_{\mathrm{col}}^n
 =R_{\mathrm{tab}}^n.
\]
Lemma~\ref{lem:column-reconstruction} records the complete boundary-data reconstruction behind this estimate. The identity condition in Case Q is essential for the perimeter bound: without equality of old and new centres, a zero-head column would have exponentially many successor words, and iteration would give an area-exponential rather than perimeter-exponential bound.  Finally, \(\Theta_\alpha\) is undefined for a move onto a sentinel, which proves sentinel consistency.
\end{proof}

\begin{lem}[Reconstruction of open successor columns]\label{lem:column-reconstruction}
Fix the east-colour word of a vertically matching column in the Wang tableau constructed above. Such a column is uniquely determined by at most two initial order bits, four new endpoint or ghost symbols, and four arrow-token slots. In particular, its number of possible completions is at most
\(4s_{\mathrm{cell}}^4(d_{\mathrm{arr}}+1)^4=K_{\mathrm{col}}\), independently of the column height.
\end{lem}

\begin{proof}
The east colours determine all old radius-one windows. Their overlaps fix the old symbols throughout the column and the old endpoint data that they expose. Each vertically matched order-bit chain is nondecreasing, and its only permitted change is the transition from $0$ to $1$ at a head centre. Hence the old chain is fixed by its south bit and its old symbols. The new chain is similarly fixed once its south bit and the new symbols have been recovered.

The reconstruction of the new centres is local except at the old head, whose uniqueness follows from the old order chain. All possible arrangements are covered by the following cases. If an old head is strictly inside the visible column, Case H fixes its outgoing labelled arrow by the recorded head token; the neighbouring Case A and Case B rules implement the unique resulting update, while Case Q leaves all more distant centres unchanged. If a head touches the first or last visible radius-one window, its possible influence is confined to the boundary row and is recorded in the endpoint or ghost symbols and exterior token slots. If no old head is visible, Case Q fixes every centre whose old window is free of a head; only those bounded boundary windows that may involve an exterior head require the recorded boundary symbols. In every case, there is no independent choice at each additional interior row.

The overlap constraints then recover the new radius-one windows, the vertical colour data, and the order chains. The arrow token at a visible head records the actual labelled transition, so distinct same-effect arrows are counted separately. Conversely, in Case B, different arrow witnesses giving an identical local tile datum are identified as a single tile type. Thus two completing columns with the same prescribed east word and the same bounded boundary data coincide tile by tile. There are at most $2^2$ choices for the two initial order bits, $s_{\mathrm{cell}}^4$ choices for the four endpoint and ghost symbols, and $(d_{\mathrm{arr}}+1)^4$ choices for the token slots, proving the bound.
\end{proof}

\section{Counting Open-Square Wang Tilings}\label{sec:wang}

The tableau correspondence in the preceding section uses prescribed boundary data, while the target counting problem has open boundaries. A side-activity amplifier assigns a dominant multiplicity to the intended boundary class. Coupling it to a ripple counter and an edge scanner recovers the initial configuration from the side length, and the same construction supplies a uniform bound for all remaining boundary classes. The resulting perimeter estimate for arbitrary rectangles will be used in the geometric compiler.

\begin{thm}[Boundary amplification]\label{thm:amplifier}
For every integer $Q\ge1$, there is a fixed Wang system with four side-activity tracks such that, on an open $n$-square for $n\ge2$, the all-active boundary class has exactly $Q^{4n-4}$ tilings and each proper side-activity class has at most $Q^{3n-2}$ tilings. A private-colour dummy tile can be adjoined, contributing exactly one additional tiling.
\end{thm}

\begin{proof}
Choose pairwise distinct atomic labels
\[
 g_W,g_E,g_S,g_N,\ a_S,a_N,\ g_{Wc},g_{Ec},g_{Sv},g_{Nv},a_W,a_E
\]
and \(d^W_q,d^E_q,d^S_q,d^N_q\) for \(1\leq q\leq Q\).
Horizontal colours are distinct symbols \(H(\alpha,\beta,\gamma,\delta)\);
vertical colours are distinct symbols \(V(\alpha,\beta,\gamma,\delta)\).

A tile carries booleans \(W,E,S,N\) and indices \(q_W,q_E,q_S,q_N\).
Adjacent active sides on the same tile must use one common \(q\).  Its colours are
\[
\begin{aligned}
 W_{\mathrm{col}}&=H(w_W,g_E,s_c,n_c),&
 E_{\mathrm{col}}&=H(g_W,e_E,s_c,n_c),\\
 S_{\mathrm{col}}&=V(w_c,e_c,s_S,g_{Nv}),&
 N_{\mathrm{col}}&=V(w_c,e_c,g_{Sv},n_N),
\end{aligned}
\]
where
\[
\begin{array}{lll}
 w_W=d^W_{q_W}\text{ if }W,\quad w_W=g_W\text{ otherwise},&
 w_c=a_W\text{ if }W,\quad w_c=g_{Wc}\text{ otherwise},\\
 e_E=d^E_{q_E}\text{ if }E,\quad e_E=g_E\text{ otherwise},&
 e_c=a_E\text{ if }E,\quad e_c=g_{Ec}\text{ otherwise},\\
 s_S=d^S_{q_S}\text{ if }S,\quad s_S=g_{Sv}\text{ otherwise},&
 s_c=a_S\text{ if }S,\quad s_c=g_S\text{ otherwise},\\
 n_N=d^N_{q_N}\text{ if }N,\quad n_N=g_{Nv}\text{ otherwise},&
 n_c=a_N\text{ if }N,\quad n_c=g_N\text{ otherwise}.
\end{array}
\]
Inactive indices do not occur in any colour and are normalized to \(1\).

No east colour contains \(d^W_q\) in its first coordinate.  Hence a \(W\)-active tile can occur only on the geometric west side.  The three symmetric assertions follow for \(E,S,N\).  Along the west side, the north and south colours carry either \(a_W\) or \(g_{Wc}\); since these differ, activity is constant down the whole side.  The other three chain coordinates give the same all-or-none conclusion.

For \(n\geq2\), let \(S\subseteq\{W,E,S,N\}\) be the active-side set and write its indicator bits as \(w,e,s,t\). The number of distinct active sites is \(\Sites(n,S)=wn+en+s(n-w-e)+t(n-w-e)\). According to the active sides, this takes one of the six values \(0\), \(n\), \(2n-1\), \(2n\), \(3n-2\), and \(4n-4\). A choice of one \(q\) at each such site determines a unique tiling.  Horizontal matching follows because an internal left tile is not \(E\)-active, an internal right tile is not \(W\)-active, and the north and south chain coordinates agree along their rows.  Vertical matching is symmetric.  Conversely, dead colours expose every \(q\).  Hence the side class has exactly
\(Q^{\Sites(n,S)}\) tilings.  Summation over the sixteen subsets gives the displayed formula.

When \(n=1\), the single tile lies on all four sides. The activity inventories contribute \(1\), \(4Q\), \(4Q\), \(2Q^2\), \(4Q\), and \(Q\) for zero sides, one side, adjacent pairs, opposite pairs, triples, and all four sides, respectively.  Their sum is \(1+13Q+2Q^2\).

The west-only set contains the inactive tile and the \(Q\) purely west-active tiles.  It has one constant inactive tiling and \(Q^n\) active-west tilings.  Finally, a private-colour tile cannot meet an amplifier tile.  Connectedness of the square grid therefore makes its mode constant and adds exactly one tiling.

A raw product of four independent wires would give \(Q^{4n}\) full frames.  The shared-\(q\) rule removes one duplicate index at each corner and gives the required \(Q^{4n-4}\).
\end{proof}

We use a four-tile ripple counter.  Its tile indexed by input bit \(b\) and carry \(c\) has west colour \(B_b\), south colour \(C_c\), east colour
\(B_{b\oplus c}\), and north colour \(C_{bc}\).  West and south boundary words determine a unique tiling.  With west bits zero and south carries one, column \(x\) represents \(x\), the east boundary represents the side length, and the top carries are zero.

\begin{thm}[Fixed open-square Wang counter]\label{thm:wang-product}
There are a fixed unweighted Wang tile set $W_{\mathrm{prod}}$ and positive constants
$R_0=4R_{\mathrm{tab}}577^2$ and $Q=2R_0+1$ such that, for every $n\ge5$, the total number $a_{\mathrm{prod}}(n)$ of open $W_{\mathrm{prod}}$-tilings of the $n$-square has the decomposition
\[
a_{\mathrm{prod}}(n)=Q^{4n-4}\operatorname{Acc}_\star(n)+\Junk(n),
\]
where the residual count satisfies
\[
0\le\Junk(n)\le1+15R_0^nQ^{3n-2}.
\]
\end{thm}

\begin{proof}
Apply Theorem~\ref{thm:tableau} to \(M_\star\).  Its tableau track contributes at most \(R_{\mathrm{tab}}^n\) open \(n\)-square patterns.  The ripple track contributes \(4^n\).

We define the scanner in full.  Let
\[
\begin{split}
 \mathsf{Buf}=\{&
 \plain{\mathsf{\#_s}},\plain{\mathsf{zero_b}},
 \plain{\mathsf{one_b}},\plain{\mathsf{work_b}},\\
 &\hd{q_{\mathrm{start}}}{\mathsf{zero_b}},
 \hd{q_{\mathrm{start}}}{\mathsf{one_b}}\}.
\end{split}
\]
Its vertical colours are
\[
 S_{\mathrm{scan}}
 =\{\mathsf{bot},\mathsf{mid},\mathsf{top},\mathsf{cap}\}
 \times\mathsf{Buf}^2
 \times\{\mathsf{lh0},\mathsf{lh1}\}
 \times\{\mathsf{modeHigh},\mathsf{modeBin}\},
\]
so \(|S_{\mathrm{scan}}|=576\).
There is one idle tile with private colour \(\mathsf{scanIdle}\) on every side.
For every \(s,t\in S_{\mathrm{scan}}\), there is a marked tile with west colour
\(\mathsf{scanIdle}\), east colour \(\mathsf{scanDead}\), south colour \(s\), and north colour \(t\).  No west colour equals \(\mathsf{scanDead}\).
Thus marked tiles occur only in the east column.  Vertical colours prevent a marked tile from meeting an idle tile vertically, so either the scanner is entirely idle or its east column is entirely marked.  Its open count satisfies \(1+576^{n+1}\leq577^{2n}\). The uncoupled counter, tableau, and scanner product therefore has at most \(R_0^n\) open tilings.

For \(b\in\bits\) and a mode \(m\), define
\[
\begin{aligned}
 \NewMid(1,m)&=\plain{\mathsf{one_b}}, &
 \NewMid(0,\mathsf{modeHigh})&=\plain{\mathsf{work_b}},\\
 \NewMid(0,\mathsf{modeBin})&=\plain{\mathsf{zero_b}}, &
 \ModeOut(0,\mathsf{modeHigh})&=\mathsf{modeHigh},\\
 \ModeOut(1,\mathsf{modeHigh})&=\mathsf{modeBin}, &
 \ModeOut(b,\mathsf{modeBin})&=\mathsf{modeBin}.
\end{aligned}
\]
Write \(s=(\phi,u,v,\lambda,m)\) and \(t=(\phi',u',v',\lambda',m')\). The relation \(\LegalScan(s,t,b)\) first requires \(u'=v\), and then exactly one of the following clauses.
\begin{enumerate}[leftmargin=4em,labelsep=.6em]
\item[\textup{Bot}.]
\[
\begin{gathered}
 \phi=\mathsf{bot},\quad\lambda=\mathsf{lh0},\quad
 u=v=\plain{\mathsf{\#_s}},\quad
 \phi'=\mathsf{mid},\quad\lambda'=\mathsf{lh0},\quad
 m=\mathsf{modeBin},\\
 v'=\hd{q_{\mathrm{start}}}{\mathsf{zero_b}}\text{ if }b=0,\qquad
 v'=\hd{q_{\mathrm{start}}}{\mathsf{one_b}}\text{ if }b=1.
\end{gathered}
\]
\item[\textup{Mid}.]
\[
 \phi=\phi'=\mathsf{mid},\quad
 \lambda=\lambda'=\mathsf{lh0},\quad
 v'=\NewMid(b,m'),\quad m=\ModeOut(b,m').
\]
\item[\textup{PreTop}.]
\[
\begin{gathered}
 \phi=\mathsf{mid},\quad\phi'=\mathsf{top},\quad
 \lambda=\mathsf{lh0},\quad\lambda'=\mathsf{lh1},\\
 v'=\plain{\mathsf{\#_s}},\quad
 m=m'=\mathsf{modeHigh},\quad b=0.
\end{gathered}
\]
\item[\textup{Top}.]
\[
\begin{gathered}
 \phi=\mathsf{top},\quad\phi'=\mathsf{cap},\quad
 \lambda=\mathsf{lh1},\quad\lambda'=\mathsf{lh0},\\
 m=m'=\mathsf{modeHigh},\quad
 v=v'=u'=\plain{\mathsf{\#_s}}.
\end{gathered}
\]
\end{enumerate}
No clause has south phase \(\mathsf{cap}\), and no clause has north phase
\(\mathsf{bot}\).  Define \(\Emit(s,t,b)=(u,v,v')\). The lookahead and mode are scanner-colour coordinates, not amplifier activity bits.

Take the product of the amplifier at \(Q\), the ripple counter, the tableau, and the scanner.  Retain a product tile only when the following cell-local couplings hold.
\begin{enumerate}
\item A west-active amplifier cell has ripple west colour \(B_0\), an accepting tableau west colour, and an idle scanner.
\item An east-active cell has a marked scanner tile, reads its ripple east bit \(b\), satisfies \(\LegalScan(s,t,b)\), and has tableau east colour
\(\mathsf{hcol}(\Emit(s,t,b))\).
\item A south-active cell has ripple south colour \(C_1\) and tableau south sentinel colour.  At the southeast corner its scanner south colour is \((\mathsf{bot},\plain{\mathsf{\#_s}},\plain{\mathsf{\#_s}},\mathsf{lh0},\mathsf{modeBin})\).
\item A north-active cell has ripple north colour \(C_0\) and tableau north sentinel colour.  At the northeast corner its scanner south colour has phase
\(\mathsf{top}\), centre \(\plain{\mathsf{\#_s}}\), lookahead
\(\mathsf{lh1}\), and mode \(\mathsf{modeHigh}\).
\item Every non-east-active cell uses the idle scanner.
\end{enumerate}
Adjoin a single all-private-colour dummy tile.  Figure~\ref{fig:wang-boundary} summarizes the boundary data in the full-frame class and the role of the unique marked scanner column.

In a full frame, the amplifier forces all four active sides.  The ripple boundary conditions give the unique incrementing counter and put the bits of \(n\) on the east edge.  The scanner diode confines marked tiles to that edge.  The bottom constraint forces Bot at row zero.  The top constraint forces Top at row \(n-1\); since no legal tile has south phase \(\mathsf{cap}\), no earlier Top is possible.  Row \(n-2\) is therefore PreTop.  Its required bit is zero because the omitted bit \(n-2\) vanishes for \(n\geq5\).  All intervening rows are Mid.

For \(j\geq0\), the \(j\)-th binary digit of \(n\) is \(\operatorname{bit}_n(j)=\lfloor n/2^j\rfloor\bmod 2\), with the least significant digit at \(j=0\).  The mode is now unique by downward recursion.  At height \(y\),
\[
 m(y)=\mathsf{modeHigh}
 \quad\Longleftrightarrow\quad
 \operatorname{bit}_n(j)=0
 \text{ for every }y\leq j\leq n-3.
\]
This follows directly from \(\ModeOut\), starting with the high mode at the top.  Therefore a zero above the leading one is emitted as
\(\mathsf{work_b}\), a lower zero as \(\mathsf{zero_b}\), and a one as
\(\mathsf{one_b}\).  Bot installs the start head at index \(1\), while PreTop and Top install the two sentinels.  PreTop explicitly checks the omitted bit \(n-2\); the ripple word already supplies the second omitted bit \(n-1=0\), so Top does not read it.  The overlap equation \(u'=v\) gives
\[
 \Emit_y=
 \bigl(\Cfg^\star_n(y-1),\Cfg^\star_n(y),\Cfg^\star_n(y+1)\bigr).
\]
The tableau bijection of Theorem~\ref{thm:tableau} then identifies its tilings with accepting branches of \(M_\star\).  Independently, the amplifier supplies one \(q\) per one of its \(4n-4\) boundary sites.  This proves the full-frame bijection.

For a proper side set, Theorem~\ref{thm:amplifier} gives at most
\(Q^{3n-2}\) amplifier tilings.  Dropping every coupling permits at most
\(R_0^n\) computational-track tilings, including every open scanner, counter, and tableau defect.  Thus each proper class has at most
\(R_0^nQ^{3n-2}\) tilings.  The private dummy colour cannot meet a product colour on the connected square grid, so it contributes one isolated tiling.  Summing the fifteen proper classes proves the junk estimate.
\end{proof}

\begin{figure}[tbp]
  \centering
  \includegraphics[width=.99\linewidth]{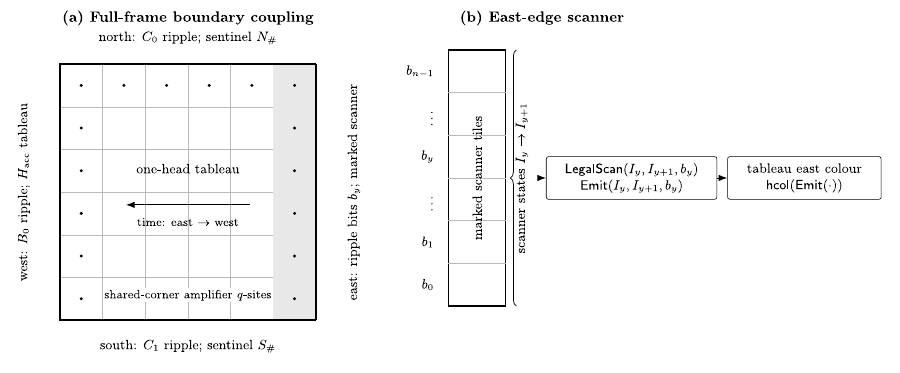}
  \caption{Full-frame boundary coupling and east-edge scanning at the Wang level.  In panel (a), the west, south, and north sides carry the fixed ripple/tableau data \(B_0,H_{\mathrm{acc}}\), \(C_1,S_\#\), and \(C_0,N_\#\), respectively.  The dots mark the shared-corner amplifier \(q\)-sites, and the shaded east column receives the ripple bits \(b_y=\operatorname{bit}_n(y)\).  Time in the one-head tableau runs from east to west.  Panel (b) isolates one scanner step: adjacent scanner states \(I_y,I_{y+1}\), together with \(b_y\), satisfy \(\LegalScan\) and determine \(\Emit(I_y,I_{y+1},b_y)\), which is exported as the tableau east colour.  Over the full marked column, these emitted colours are exactly the radius-one windows of \(\Cfg_n^\star\).}
  \label{fig:wang-boundary}
\end{figure}

\begin{thm}[Perimeter bound on arbitrary rectangles]\label{thm:rectangle-bound}
Let \(a_{\mathrm{rect}}(w,h)\) count open \(W_{\mathrm{prod}}\)-tilings of
\(\Rect_{w,h}\).  There is a fixed integer \(R_{\mathrm{rect}}\geq2\) such that
\[
 a_{\mathrm{rect}}(w,h)\leq R_{\mathrm{rect}}^{w+h}
 \qquad(w,h\geq1).
\]
One may take
\[
 R_{\mathrm{rect}}
 =1+16Q^2\cdot2R_{\mathrm{tab}}
       (|S_{\mathrm{scan}}|+1)^2.
\]
The bound also applies when \(w=1\) or \(h=1\).
\end{thm}

\begin{proof}
The dummy colour again gives one isolated constant tiling.  Dropping all couplings bounds every other tiling by the product of four independent track counts.

The ripple track is determined by \(h\) west bits and \(w\) south carries, so it has exactly \(2^{w+h}\) tilings.  For the tableau track, there are
\((s_{\mathrm{cell}}^3)^h\) east words and at most
\(K_{\mathrm{col}}^w\) westward column extensions.  Since
\(R_{\mathrm{tab}}=s_{\mathrm{cell}}^3K_{\mathrm{col}}\),
\[
 a_{\mathrm{tab}}(w,h)\leq R_{\mathrm{tab}}^{w+h}.
\]
The scanner is either idle or has one marked east column, giving \(a_{\mathrm{scan}}(w,h)=1+576^{h+1}\leq577^{2(w+h)}\). For the amplifier, dead colours and chain colours still force an all-or-none subset of the four geometric sides.  Charge one independent \(q\) for every active side cell before imposing shared-corner equalities.  A chosen side set therefore has at most
\(Q^{2w+2h}\)
tilings, and summing sixteen sets gives
\(16Q^{2(w+h)}\).  This charging remains valid when opposite sides occupy the same physical row or column.  In particular, for \(w=1\), west and east indices may remain independent; for \(w=h=1\), the exact inventory
\(1+13Q+2Q^2\) is below \(16Q^4\).

Multiplying the four bounds and absorbing the factor \(16\) into the exponent yields
\[
 a_{\mathrm{rect}}(w,h)
 \leq1+
 \bigl(16Q^2\cdot2R_{\mathrm{tab}}\cdot577^2\bigr)^{w+h}
 \leq R_{\mathrm{rect}}^{w+h}.
\]
\end{proof}

\section{Rectangular Decompositions of the Macrogrid}\label{sec:dissections}

Closed interfaces in the geometric realization may divide the macrogrid into several components. The analysis requires two combinatorial facts: every nontrivial rectangular dissection has a sufficiently long seam and a controlled number of possible seam configurations, and the saturation condition enforced by the carrier interfaces makes each connected component a filled rectangle. Together, these facts classify the resulting tilings and control the multiplicity of nontrivial decompositions.

\begin{thm}[Rectangular-dissection bound]\label{thm:dissections}
Let a rectangular dissection of \(\Square_n\) be a partition into at least two nonempty discrete rectangles.  A seam pair is a pair of adjacent cells in different blocks, and its seam segment is their common unit side.  Form the seam graph whose vertices are seam segments and whose edges join segments sharing an endpoint.  If \(I\) is the number of seam pairs, then:
\begin{enumerate}
\item \(I\geq n\);
\item every connected component of the seam graph contains at least \(n\) segments;
\item for fixed \(I\geq1\), the number of rectangular dissections with \(I\) seam pairs is at most \(256^I\).
\end{enumerate}
\end{thm}

\begin{proof}
Every internal side of a rectangular block consists of seam segments, and every seam segment lies on an internal side of both adjacent blocks.

Let \(f(x)\) be the number of blocks meeting column \(x\), and \(g(y)\) the number meeting row \(y\).  Each column is partitioned into \(f(x)\) intervals, while each row is partitioned into \(g(y)\) intervals.  Hence
\[
 I=\sum_{x=0}^{n-1}(f(x)-1)
   +\sum_{y=0}^{n-1}(g(y)-1).
\]
If every \(f(x)\geq2\), the first sum is at least \(n\).  Otherwise some block contains a whole column.  Such a rectangular block has height \(n\); unless it is the entire square, one of its vertical sides is internal and contributes \(n\) seams.  The entire-square case is excluded by nontriviality.  This proves \(I\geq n\).

Consider the internal sides of one block in cyclic order.  Adjacent internal sides share a corner and therefore belong to the same seam component.  If its internal sides do not all lie in one component, then it has either one internal side or two opposite internal sides.  In either case every such side has length \(n\).  Thus each block satisfies the dichotomy:
\[
 \text{every internal side has length }n,
 \quad\text{or all its internal sides lie in one seam component.}
\]
Suppose a seam component \(C\) has fewer than \(n\) segments.  Any block incident with \(C\) must satisfy the second alternative and hence has all internal sides in \(C\).  A neighbouring block is incident along one of these sides and has the same property.  The adjacency graph of blocks is connected, since a row-then-column cell path between two blocks crosses a sequence of neighbouring blocks.  Therefore every block is incident with \(C\), and every seam belongs to \(C\).  This would give
\(|C|=I\geq n\), a contradiction.

A dissection is determined by its seam set: delete the seam adjacencies from the cell grid, and the connected components are precisely its rectangular blocks.  It remains to count possible seam sets.

Let \(\mathcal C_m\) be the connected sets of \(m\) unit grid segments with endpoints in \(\{0,\ldots,n\}^2\).  Double every edge of such a set.  The resulting connected even-degree multigraph has an Euler tour of length \(2m\).  An initial endpoint and the sequence of four possible unit steps determine the segment set, giving \(|\mathcal C_m|\leq(n+1)^2 4^{2m}=(n+1)^2 16^m\). The elementary inequality \((n+1)^2\leq4^n\) holds for all \(n\geq1\).

An admissible seam set with \(I\) segments is a disjoint union of \(r\) connected components of sizes \(m_1,\ldots,m_r\), where each \(m_j\geq n\) and
\(\sum m_j=I\).  Thus \(r\leq I/n\). For each ordered size tuple, \(\prod_j|\mathcal C_{m_j}|\leq(n+1)^{2r}16^I\leq4^I16^I\). The number of ordered positive compositions, allowing every
\(1\leq r\leq I\), is at most \(4^I\): for fixed \(r\), partial sums give at most \(2^I\) choices, and there are at most \(I\leq2^I\) choices of \(r\).
Consequently, there are at most \(4^I\cdot4^I\cdot16^I=2^{8I}=256^I\) admissible seam sets, and hence at most this many dissections.
\end{proof}

We next isolate the local graph principle used by the geometric compiler. Fix \(P_{\mathrm{mac}}\geq1\) and \(c_{\mathrm{col}}\geq0\), and write \(L_{\mathrm{mac}}(n)=P_{\mathrm{mac}}n+2c_{\mathrm{col}}\). Define
\[
 x_0=0,\quad x_n=L_{\mathrm{mac}}(n),\quad
 x_i=c_{\mathrm{col}}+P_{\mathrm{mac}}i\quad(1\leq i<n).
\]
Let
\[
 \mathsf{Mac}_{p,q}
 =\{x_p,\ldots,x_{p+1}-1\}
  \times\{x_q,\ldots,x_{q+1}-1\}.
\]
For a tiling \(\Pi\), open the macro-edge between two adjacent macrocells when one placement contains adjacent physical cells on its two sides.  This gives a spanning subgraph \(G_{\mathrm{mac}}(\Pi)\) of the \(n\)-square index grid.

A grid subgraph is saturated when, in every \(2\)-by-\(2\) index square, two consecutive present edges force all four edges.

\begin{thm}[Saturation and rectangular macroblocks]\label{thm:saturation}
The following assertions hold.
\begin{enumerate}
\item Every connected component of a saturated index graph is a filled combinatorial rectangle: all internal horizontal and vertical index adjacencies of the rectangle are present.
\item The macrocells \(\mathsf{Mac}_{p,q}\) are nonempty, pairwise disjoint, and partition \(\Square_{L_{\mathrm{mac}}(n)}\).
\item If every prototile is 4-connected, then the macrocells met by one placement lie in a single component of \(G_{\mathrm{mac}}(\Pi)\).
\item If \(G_{\mathrm{mac}}(\Pi)\) is saturated, its components induce a partition of the physical square into nonempty discrete rectangles.  If there is more than one component, this is a rectangular dissection.  No closed macro-adjacency lies inside one component.
\item Let \(d_T\) be the maximum \(\ell_\infty\)-diameter of a prototile.  If saturation fails in a \(2\)-by-\(2\) macro-window \(Q_{\mathrm{win}}\), every placement meeting the window is contained in its \(\ell_\infty\)-neighbourhood \(Q_{\mathrm{ext}}\) of radius \(d_T\).  Those placements determine the four offending edge bits, and \(Q_{\mathrm{ext}}\) has width and height at most
\[
 2(P_{\mathrm{mac}}+2c_{\mathrm{col}})+2d_T.
\]
\end{enumerate}
\end{thm}

\begin{proof}
Let \(C\) be a connected component of a saturated index graph.  Choose a filled combinatorial rectangle
\[
 R=\{a,\ldots,a'\}\times\{b,\ldots,b'\}\subseteq C
\]
of maximum cardinality.  Suppose an edge of \(C\) leaves its north side at
\((x_0,b')\).  Let \(S\) be the set of \(x\) for which the vertical edge from
\((x,b')\) to \((x,b'+1)\) is present.  If \(x\in S\) and \(x<a'\), the vertical edge at \(x\) and the filled horizontal edge below it are consecutive sides of a \(2\)-by-\(2\) square.  Saturation supplies the upper horizontal edge and the vertical edge at \(x+1\).  The symmetric argument propagates left.  Hence
\(S=\{a,\ldots,a'\}\), and the full added row is filled.  This contradicts maximality.  Coordinate swaps and reflections treat the other three sides.  Therefore \(R=C\).

The sequence \(x_0,\ldots,x_n\) is strictly increasing.  Its endpoint increments are \(P_{\mathrm{mac}}+c_{\mathrm{col}}\), its interior increments are
\(P_{\mathrm{mac}}\), and when \(n=1\) its sole increment is
\(P_{\mathrm{mac}}+2c_{\mathrm{col}}\).  The induced coordinate intervals therefore partition
\(\{0,\ldots,L_{\mathrm{mac}}(n)-1\}\), proving the macrocell assertion.

Let a 4-connected placement meet two macrocells.  A cell path inside the placement crosses between either the same macrocell or two adjacent macrocells.  Every genuine crossing witnesses an open macro-edge.  The visited index cells thus form a walk in \(G_{\mathrm{mac}}(\Pi)\), proving containment in one component.

Under saturation, each component is a filled index rectangle.  The union of its consecutive coordinate intervals is \(\{x_a,\ldots,x_{a'+1}-1\}\times\{x_b,\ldots,x_{b'+1}-1\}\), a nonempty discrete rectangle.  Components partition the index grid, so these unions partition the physical square.  Filledness says every internal adjacency is open; conversely a closed adjacency separates components.

Finally, if a placement meets \(Q_{\mathrm{win}}\), every one of its cells lies within \(\ell_\infty\)-distance \(d_T\) of the meeting cell.  It is therefore contained in \(Q_{\mathrm{ext}}\).  A macrocell side length is at most
\(P_{\mathrm{mac}}+2c_{\mathrm{col}}\), so the two-by-two window has width and height at most twice that quantity; thickening adds \(2d_T\).  The placements meeting the window contain every possible witness for its four open bits, so they determine the saturation failure.
\end{proof}

\section{The Polyomino Construction}\label{sec:geometry}

The fixed Wang system is realized geometrically by a family of connected polyominoes. Each carrier stores the Wang tile identity in an interior cavity and encodes its edge colours and opening data in boundary registers. The proof establishes unique macrogrid anchoring and an exact correspondence between tilings and rectangular Wang-tiled blocks. The dissection estimates then isolate the contribution of the single-block configuration.

Let \(W=W_{\mathrm{prod}}\), \(M_W=|W|\), and for \(u\in W\) write \(W(u),E(u),S(u),N(u)\) for its west, east, south, and north colours.  Let
\(\operatorname{Col}_H,\operatorname{Col}_V\) be its horizontal and vertical colour sets, of cardinalities \(H_{\mathrm{col}}\) and \(V_{\mathrm{col}}\).  Fix bijections
\[
 \jmath\colon W\to\{0,\ldots,M_W-1\},\quad
 \iota_H\colon\operatorname{Col}_H\to\{0,\ldots,H_{\mathrm{col}}-1\},\quad
 \iota_V\colon\operatorname{Col}_V\to\{0,\ldots,V_{\mathrm{col}}-1\}.
\]
Define
\[
\begin{split}
 s_{\mathrm{amp}}&=1024R_{\mathrm{rect}}^3,\\
 K_s&=4s_{\mathrm{amp}}\max(H_{\mathrm{col}},V_{\mathrm{col}}),\\
 P_s&=K_s+4M_W+64,\qquad b_\star=P_s-8.
\end{split}
\]
For \(a,b\in\bits\), a horizontal colour \(\chi\), and
\(0\leq\ell<s_{\mathrm{amp}}\), put
\[
 A_H(a,b,\chi,\ell)
 =8+4(s_{\mathrm{amp}}\iota_H(\chi)+\ell)+2a+b,
\]
and define \(A_V\) analogously using \(\iota_V\).  Every code lies in
\(\{8,\ldots,K_s+7\}\), strictly below \(b_\star\), and the four arguments are recovered from it.

Let
\[
 \Lplug=\{(0,1),(1,0),(1,1)\},\qquad
 q_u=(16+4\jmath(u),16),\qquad
 \Cav_u=q_u+\Lplug.
\]
For an opening vector
\(\eta=(\eta_W,\eta_E,\eta_S,\eta_N)\in\bits^4\), let
\(\ell_D=0\) whenever \(\eta_D=0\).  Starting with
\(B_s=\{0,\ldots,P_s-1\}^2\), delete \(\Cav_u\).  On each open side also perform the following dent-and-bump operation:
\[
\begin{array}{c|c|c}
\text{side}&\text{dent}&\text{bump}\\ \hline
W&(0,A_H(\eta_S,\eta_N,W(u),\ell_W))&(-1,b_\star)\\
E&(P_s-1,b_\star)&(P_s,A_H(\eta_S,\eta_N,E(u),\ell_E))\\
S&(A_V(\eta_W,\eta_E,S(u),\ell_S),0)&(b_\star,-1)\\
N&(b_\star,P_s-1)&(A_V(\eta_W,\eta_E,N(u),\ell_N),P_s).
\end{array}
\]
Denote the resulting support by \(\Car(\eta,u,\ell)\), and let
\(\Tstar\) consist of \(\Lplug\) and every such carrier.  Figure~\ref{fig:carrier-interface} shows the tile-specific cavity and the reciprocal interface registers.

\begin{figure}[!htbp]
  \centering
  \includegraphics[width=.89\linewidth]{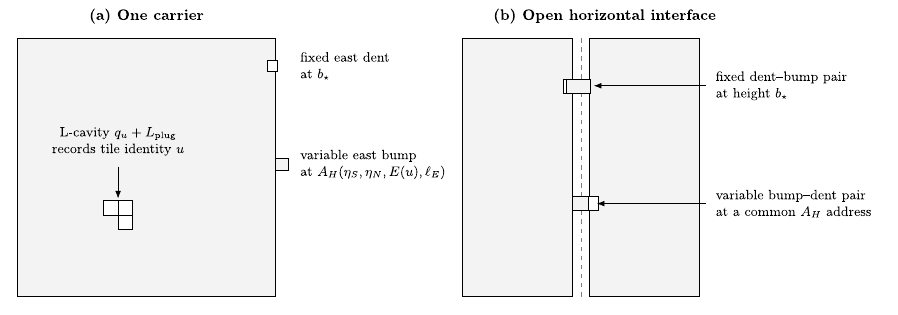}
  \caption{A connected carrier and the geometry of a matching interface. The interior cavity records the identity of the simulated Wang tile. A fixed and a variable dent--bump pair encode the state of an open interface. Matching addresses force agreement of the adjacent local data.}
  \label{fig:carrier-interface}
\end{figure}

\begin{lem}[Local geometric rigidity]\label{lem:carrier-rigidity}
Every carrier \(\Car(\eta,u,\ell)\) satisfies the following two properties.
\begin{enumerate}
\item In any relative row or column of the body square \(B_s\), no three consecutive body cells are all absent from the carrier.
\item Relative rows and columns \(1\) and \(2\), and each three-by-three corner of \(B_s\), are intact.
\end{enumerate}
\end{lem}

\begin{proof}
The interior cavity is contained in a two-by-two box and removes at most two consecutive cells from any row or column.  Every side dent removes a single boundary cell.  On any fixed body row, boundary dents can occur only at the two extreme columns; on any fixed body column, they can occur only at the two extreme rows.  The variable dent coordinates lie in \(\{8,\ldots,K_s+7\}\), the fixed coordinate is \(b_\star=P_s-8\), and the cavity lies in the strip \(\{16,\ldots,4M_W+13\}\times\{16,17\}\).  If a boundary dent shares a row or column with the cavity, at least fifteen occupied body cells separate them.  Hence neither the cavity nor any combination of boundary dents creates three consecutive missing body cells.  The same coordinate separations place every cavity and dent away from relative rows and columns \(1,2\) and away from all three-by-three body corners.
\end{proof}

For \(n\geq1\), let \(\Part_n\) denote the set of rectangular partitions of the index square \(\{0,\ldots,n-1\}^2\), including the one-block partition.  For \(D\in\Part_n\), let \(I(D)\) be its number of seam adjacencies, and for a block \(B\in D\) let \(w_B,h_B\) denote its width and height.

\begin{thm}[Geometric compilation and counting identity]\label{thm:geometry}
The fixed translation-only family \(\Tstar\) has the following properties.
\begin{enumerate}
\item Every carrier is 4-connected, has area \(P_s^2-3\), lies in
\(\{-1,\ldots,P_s\}^2\), and the map
\((\eta,u,\ell)\mapsto\Car(\eta,u,\ell)\) is injective on unlabelled supports.  The plug is 4-connected of area \(3\).
\item For every \(n\geq1\), every \(\Tstar\)-tiling of \(\Square_{P_sn}\) has exactly one carrier anchored at each \((iP_s,jP_s)\), \(0\leq i,j<n\).
\item Adjacent anchored carriers are disjoint exactly when their reciprocal opening bits agree and, whenever the common opening bit is one, the corresponding perpendicular opening bits, Wang colours, and amplifier labels also agree.  Diagonal carriers are disjoint, and outer opening bits vanish.
\item For every \(n\geq1\), after the anchored carriers are placed, the residual set is the disjoint union of the \(n^2\) translated cavities.  Each cavity has one plug placement, and no additional carrier fits.
\item For every \(n\geq1\), tilings of \(\Square_{P_sn}\) are in bijection with a partition \(D\in\Part_n\), an open \(W\)-tiling of every block, and one label in
\(\{0,\ldots,s_{\mathrm{amp}}-1\}\) on every open macro-edge.
\item Let \(E_n=2n(n-1)\). Then
\[
 g_{\Tstar}(P_sn)
 =\sum_{D\in\Part_n}
 s_{\mathrm{amp}}^{E_n-I(D)}
 \prod_{B\in D}a_{\mathrm{rect}}(w_B,h_B).
\]
For \(n=1\) this equals \(a_{\mathrm{rect}}(1,1)\).  For \(n\geq2\), the nontrivial-partition contribution \(J(n)\) satisfies
\[
 0\leq\frac{J(n)}{s_{\mathrm{amp}}^{E_n}}
 \leq\frac43\,4^{-n}<1,
\]
and consequently
\[
 \left\lfloor
 \frac{g_{\Tstar}(P_sn)}
      {s_{\mathrm{amp}}^{E_n}}
 \right\rfloor
 =a_{\mathrm{rect}}(n,n).
\]

\end{enumerate}
\end{thm}

\begin{proof}
The coding maps \(A_H,A_V\) are injective because their residues modulo four recover the two perpendicular bits, while quotient and remainder by
\(s_{\mathrm{amp}}\) recover the colour and label.  The cavity lies in the interior strip \(\{16,\ldots,4M_W+13\}\times\{16,17\}\), far from every edge feature.  Each open side deletes one cell and adds one cell, so a carrier has area \(P_s^2-3\).

First delete the entire two-by-two box containing the cavity from the inner rectangle
\(\{1,\ldots,P_s-2\}^2\); the remainder is 4-connected.  Reinsert the southwest cell of that box, which is the unique cell of the box not removed by the L-shaped cavity.  It is adjacent to the surrounding remainder and therefore preserves 4-connectivity.  Every undeleted boundary cell attaches inward.  Each bump attaches to an undeleted boundary cell because all variable code positions lie below \(b_\star\).  Hence every carrier is 4-connected.

The cell \((0,0)\) uniquely minimizes \(x+y\) in every carrier.  Equality up to translation therefore forces equal origins.  The presence or absence of the four side features recovers \(\eta\), the interior cavity recovers \(u\), and on every open side the variable address---carried by the dent on the west and south sides and by the bump on the east and north sides---recovers its label.  This proves injectivity.  It also proves the identity test for equal-colour Wang tiles.

We establish anchoring by row-major induction, increasing \(j\) first and, within a fixed macro-row \(j\), increasing \(i\).  The repeated rigidity input is Lemma~\ref{lem:carrier-rigidity}: an anchored carrier cannot expose three consecutive body cells in one row or column, and the first two body rows and columns and all three-by-three body corners are intact.

Let \(z=(iP_s,jP_s)\) be the next macro-corner.  A translate of \(\Lplug\) covering \(z\) has origin \((iP_s,jP_s-1)\), \((iP_s-1,jP_s)\), or \((iP_s-1,jP_s-1)\). At the outer boundary one of its cells leaves the square.  Otherwise it overlaps the intact cell \((1,P_s-1)\) of the lower carrier, the intact cell
\((P_s-1,1)\) of the left carrier, or an intact corner cell.  Hence the placement covering \(z\) is a carrier.

Write its origin as \((a,b)\).  The intact body corners and containment in \(\Square_{P_sn}\) imply \(0\leq a,b\leq P_sn-P_s\). Because the carrier covers \(z\) and its relative support is contained in \(\{-1,\ldots,P_s\}^2\), its origin satisfies
\[
iP_s-P_s\leq a\leq iP_s+1,\qquad jP_s-P_s\leq b\leq jP_s+1.
\]
If \(b\leq jP_s-2\), then \(j\geq1\) and the new carrier's intact relative row \(1\) lies inside the already anchored macro-row \(j-1\).  By the horizontal-origin bounds, its length-\(P_s\) body interval meets at least one anchored macrocell of that row in three consecutive positions.  Those three positions are occupied by the new carrier, while the no-three-gap property prevents the anchored carrier from omitting all three; hence the two placements overlap.  If \(b=jP_s-1\), again \(j\geq1\).  The intact cells with relative coordinates \((0,0),(1,0),(2,0)\) of the new carrier lie on the top body row of the carrier directly below.  On that top row the only possible boundary omission is the single north dent, and the corner cells are intact, so at least one of these three cells is already occupied.  Both cases are impossible.

If \(b=jP_s+1\), the corner \(z\) must be the south bump, so \(a=iP_s-b_\star=(i-1)P_s+8\). For \(i=0\), the body leaves the square.  For \(i\geq1\), its body overlaps the left carrier throughout
\[
\{(i-1)P_s+8,\ldots,iP_s-1\}\times\{jP_s+1,\ldots,jP_s+P_s-1\}.
\]
In particular, \(((i-1)P_s+9,jP_s+2)\) belongs to both carriers: it is outside both cavities and all possible dents.  Thus \(b=jP_s\).  With \(b=jP_s\) fixed, if \(a\leq iP_s-2\) then necessarily \(i\geq1\), and the new carrier's intact relative column \(1\) meets the already anchored carrier immediately to the left in at least three consecutive vertical body positions.  The no-three-gap property forces an overlap.  If \(a=iP_s-1\), the new carrier's intact southwest corner overlaps the left carrier's intact southeast corner.  Finally, if \(a=iP_s+1\), covering \(z\) would require a relative cell with first coordinate \(-1\) at height zero, whereas the only such cell is the west bump at height \(b_\star\).  Hence \(a=iP_s\), and therefore
\((a,b)=(iP_s,jP_s)\).  Induction anchors exactly one carrier at every macro-origin. Table~\ref{tab:anchoring-cases} records the exhaustive alternatives and their overlap certificates.

\begin{table}[!htbp]
\centering
\small
\renewcommand{\arraystretch}{1.12}
\caption{Exhaustive alternatives at the next row-major macro-corner $z=(iP_s,jP_s)$. Here $(a,b)$ denotes the origin of a carrier that covers $z$. Each rejected placement either crosses the physical boundary or overlaps a previously anchored carrier.}
\label{tab:anchoring-cases}
\begin{tabularx}{\linewidth}{@{}>{\raggedright\arraybackslash}p{.285\linewidth}X@{}}
\toprule
Placement covering $z$ & Exclusion or surviving case \\
\midrule
Plug at $z+(0,-1)$ & Crosses the lower boundary or overlaps the intact lower carrier. \\
Plug at $z+(-1,0)$ & Crosses the left boundary or overlaps the intact left carrier. \\
Plug at $z+(-1,-1)$ & Crosses the boundary or overlaps an intact previously occupied corner. \\
Carrier with $b\le jP_s-2$ & Its intact body row~$1$ forces three consecutive occupied cells in the preceding macro-row, contradicting the no-three-gap property. \\
Carrier with $b=jP_s-1$ & Three intact south body cells meet the occupied top row below, where at most one north dent is absent. \\
Carrier with $b=jP_s+1$ & This requires the south bump at $z$. It crosses the boundary if $i=0$; otherwise it overlaps at $((i-1)P_s+9,jP_s+2)$. \\
Carrier with $b=jP_s$, $a\le iP_s-2$ & Its intact body column~$1$ overlaps the anchored macro-column to the left in three consecutive cells. \\
Carrier with $b=jP_s$, $a=iP_s-1$ & The new southwest corner overlaps the left carrier's intact southeast corner. \\
Carrier with $b=jP_s$, $a=iP_s+1$ & Covering $z$ would require a cell at relative coordinate $(-1,0)$, whereas the only west bump is at height $b_\star$. \\
Carrier at $(a,b)=z$ & The unique surviving anchored position. \\
\bottomrule
\end{tabularx}
\end{table}

Consider two horizontal neighbours.  If the left east bit and right west bit differ, the existing bump overlaps an undeleted body cell of the neighbour.  Hence they agree.  If both are open, the left east bump can avoid overlap only by occupying the right west dent, so
\[
A_H(\eta_S,\eta_N,E(u),\ell_E)=A_H(\eta'_S,\eta'_N,W(u'),\ell'_W).
\]
Injectivity yields equality of the perpendicular bits, Wang colours, and labels.  The fixed bump at height \(b_\star\) reciprocally fills the left east dent.  If both sides are closed, the adjacent body columns are disjoint.  The vertical argument is identical with \(A_V\).

An open outer side would place a bump outside the square, so outer bits vanish.  Diagonal support boxes meet in four cells around a macro-corner.  The lower-left carrier occupies only its northeast body corner among those four; the upper-right carrier occupies only its southwest body corner.  The other two cells cannot be bumps because variable coordinates are at most \(K_s+7<P_s-8\), while the fixed bump coordinate is \(b_\star\neq0\).  Thus diagonal carriers are disjoint.  Boxes separated by two macro-pitches are disjoint.

All exchanged dents are filled by reciprocal bumps, and closed interfaces leave no holes.  The uncovered cells are therefore precisely the translated cavities.  They are mutually separated, 4-connected, and have area three.  Any remaining placement lies in one residual component.  A carrier is too large, while a plug fitting a three-cell component must equal it.  The translation is uniquely the cavity origin.

Set \(P_{\mathrm{mac}}=P_s\) and \(c_{\mathrm{col}}=0\).  A macro-edge is open exactly when the reciprocal carrier bits on that interface are one; plugs lie strictly inside one macrocell.  We verify saturation on one orientation, the others following by symmetry.  Suppose the bottom horizontal edge and the left vertical edge of a \(2\)-by-\(2\) macro-square are open.  The open bottom edge equates the north bits of its two incident carriers.  Since the bottom-left north bit is one (the left vertical interface is open), the bottom-right north bit is also one; reciprocal-bit equality then makes the right vertical interface open.  That open right interface equates the west bits of its two incident carriers.  The bottom-right west bit is one (the bottom interface is open), so the top-right west bit is one; reciprocal-bit equality then makes the top horizontal interface open.  Thus two consecutive open edges force all four.

Theorem~\ref{thm:saturation} makes the opening components filled rectangles.  The carrier identities on an open edge give Wang-colour matching inside each component, and the common label supplies one element of
\(\{0,\ldots,s_{\mathrm{amp}}-1\}\).  Conversely, a rectangular partition, an open Wang tiling on every block, and edge labels determine all carrier records.  The interface rules give disjointness, the diagonal check handles corner contacts, and the unique plugs fill the cavities.  These maps are inverse: a physical tiling determines the partition, the Wang tile identities and the interface labels, and conversely these data prescribe all carrier placements and the unique cavity plugs.  An all-closed singleton still has one distinct carrier for every \(u\in W\), because \(u\) is recorded by its cavity.

There are \(E_n=2n(n-1)\) index-grid edges.  A partition with \(I(D)\) seam edges has \(E_n-I(D)\) open edges, proving the exact sum.

For every rectangular partition, \(\sum_{B\in D}(w_B+h_B)=2n+I(D)\), because block perimeters count the outer boundary once and each seam twice.  Figure~\ref{fig:dissection-suppression} contrasts the canonical one-block term with a representative nontrivial rectangular partition and identifies the seam parameter \(I(D)\) that drives the suppression estimate below.

\begin{figure}[!htbp]
  \centering
  \includegraphics[width=.84\linewidth]{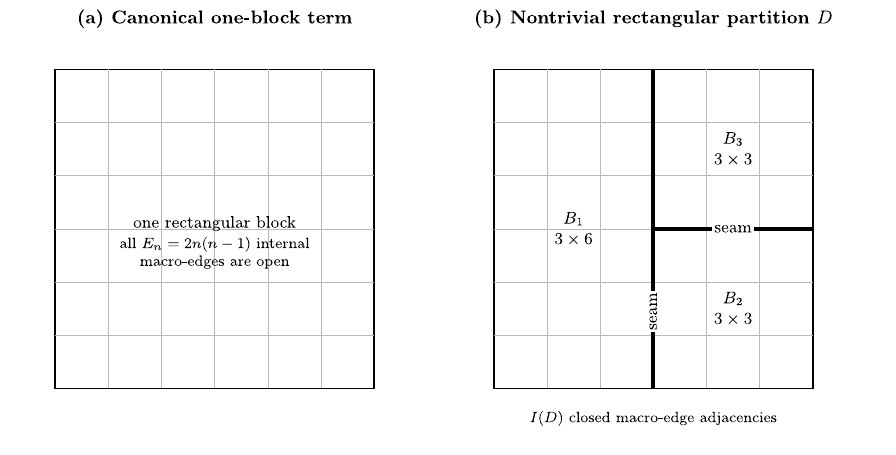}
  \caption{One-block and nontrivial rectangular decompositions of the macrogrid. Each closed adjacency on an internal seam removes one independent interface-label choice. The canonical one-block term therefore has the maximum interface multiplicity, while other decompositions are bounded by their seam length.}
  \label{fig:dissection-suppression}
\end{figure}

The bijection in Theorem~\ref{thm:geometry} assigns every geometric tiling to exactly one partition \(D\in\Part_n\).  The unique one-block partition is the canonical term; every other tiling appears in exactly one nontrivial partition.  For a fixed \(D\), the exact sum above already includes all open Wang tilings inside every block and all amplifier-label choices on every open macro-edge.  Grouping the remaining partitions by their seam count \(I(D)\) therefore accounts for every noncanonical tiling exactly once.  The two estimates below then bound, respectively, how many such partitions have a given \(I\) and how large the complete internal Wang contribution of each partition can be.

For a nontrivial partition, Theorem~\ref{thm:dissections} gives $I(D)\ge n$ and at most $256^I$ dissections with $I$ seams. By Theorem~\ref{thm:rectangle-bound} and the block-perimeter identity,
\[
\prod_{B\in D}a_{\mathrm{rect}}(w_B,h_B)\le R_{\mathrm{rect}}^{2n+I(D)}.
\]
Writing $J(n)$ for the sum of the nontrivial-partition terms, the exhaustive partition bijection now yields
\[
\begin{aligned}
\frac{J(n)}{s_{\mathrm{amp}}^{E_n}}
&\le \sum_{I=n}^{E_n}256^I s_{\mathrm{amp}}^{-I}R_{\mathrm{rect}}^{2n+I}\\
&\le \sum_{I=n}^{\infty}\left(\frac{256R_{\mathrm{rect}}^3}{s_{\mathrm{amp}}}\right)^I
 =\frac43\,4^{-n}<1,
\end{aligned}
\]
where the second inequality uses $2n+I\le 3I$ and $s_{\mathrm{amp}}=1024R_{\mathrm{rect}}^3$. This sum includes every noncanonical tiling through its unique rectangular partition.
For \(n\geq2\) this is at most \(1/12\).  The one-block term is
\(s_{\mathrm{amp}}^{E_n}a_{\mathrm{rect}}(n,n)\), so the normalized error is strictly below one and the floor identity follows.  For \(n=1\), there is only the singleton partition and \(E_1=0\).
\end{proof}

In particular, the identity cavity distinguishes Wang tile types even when their four edge colours coincide. For consecutive register values, the cavity positions are respectively
\[
 \{(16,17),(17,16),(17,17)\},\qquad
 \{(20,17),(21,16),(21,17)\}.
\]

\section{Recovering the Wang Count}\label{sec:first-readout}

The boundary-amplified Wang system separates the desired computation histories from the other open-boundary configurations by counting multiplicity. A bound on the remaining configurations makes the accepting-path count recoverable by integer division. The section also verifies the bit complexity of this readout, establishing metric completeness for the fixed Wang counter.

Define the total Wang-counting function by \(G_{\mathrm{prod}}(\binenc(n))=a_{\mathrm{prod}}(n)\) for \(n\geq1\), and set it to zero on all other strings.

\begin{thm}[Metric completeness of the Wang count]\label{thm:wang-completeness}
The function \(G_{\mathrm{prod}}\) has the following properties.
\begin{enumerate}
\item \(G_{\mathrm{prod}}\in\#\EXP\).
\item For \(n=E_{\mathrm{side}}(x)\),
\[
 a_{\mathrm{prod}}(n)
 =Q^{4n-4}(1+C_U(x))+\Junk(n).
\]
\item Let \(n_1\) be the least integer at least two satisfying
\(2^{n_1-1}>15Q^2\), and let \(n_{\mathrm{thr}}=\max(5,n_1)\).
Then
\[
 0\leq\Junk(n)<Q^{4n-4}\qquad(n\geq n_{\mathrm{thr}}).
\]
\item The exceptional set \(X_{\mathrm{exc}}=\{x:E_{\mathrm{side}}(x)<n_{\mathrm{thr}}\}\) is finite.  There is a polynomial-time finite-table function
\(H_{\mathrm{exc}}\) agreeing with \(C_U\) on this set.
\item The instance map \(r_{\mathrm{met}}(x)=\binenc(E_{\mathrm{side}}(x))\) and the postprocessor
\[
 h_\star(x,y)=
 \begin{cases}
 H_{\mathrm{exc}}(x),&E_{\mathrm{side}}(x)<n_{\mathrm{thr}},\\[2mm]
 \max\!\left(0,
 \left\lfloor\dfrac{y}
 {Q^{4E_{\mathrm{side}}(x)-4}}\right\rfloor-1\right),
 &E_{\mathrm{side}}(x)\geq n_{\mathrm{thr}},
 \end{cases}
\]
form a polynomial-time one-query metric reduction from \(C_U\) to
\(G_{\mathrm{prod}}\).  The exponentiation is stopped once its value exceeds \(y\).  On exceptional inputs the postprocessor ignores \(y\).
\item \(G_{\mathrm{prod}}\) is \(\#\EXP\)-complete under polynomial-time metric reductions.  
\end{enumerate}
\end{thm}

\begin{proof}
An open \(W_{\mathrm{prod}}\)-tiling of an \(n\)-square has a unique encoding by
\(O(n^2)\) fixed-length tile labels.  A deterministic verifier checks all horizontal and vertical adjacencies in \(O(n^2)\) time.  On a binary input of length \(\ell\), \(n<2^\ell\), so guessing the labels and checking them takes
\(2^{O(\ell)}\) time.  Malformed strings and zero are rejected deterministically.  This proves membership.

For \(n=E_{\mathrm{side}}(x)\), Theorem~\ref{thm:source-machine} gives
\(\operatorname{Acc}_\star(n)=1+C_U(x)\), and
Theorem~\ref{thm:wang-product} gives
\[
 a_{\mathrm{prod}}(n)
 =Q^{4n-4}\operatorname{Acc}_\star(n)+\Junk(n).
\]
This proves substitution.

Since \(Q=2R_0+1\),
\(0<R_0/Q<1/2\).
For \(n\geq5\),
\[
 \frac{\Junk(n)}{Q^{4n-4}}
 \leq Q^{-(4n-4)}
 +15Q^2\left(\frac{R_0}{Q}\right)^n
 <Q^{-(4n-4)}+15Q^2\,2^{-n}.
\]
When \(n\geq n_{\mathrm{thr}}\), both terms are below \(1/2\), proving the strict threshold estimate.

The lower bound \(E_{\mathrm{side}}(x)\geq2^{74|x|+74}\) tends to infinity with \(|x|\), so only finitely many strings lie below the fixed threshold.  Their values \(C_U(x)\) define a finite hard-coded table
\(H_{\mathrm{exc}}\).  A lookup in a fixed table is polynomial-time.  This existence claim does not require the threshold to be below \(2^{74}\).

For a nonexceptional input, put \(n=E_{\mathrm{side}}(x)\) and
\(B=Q^{4n-4}\).  The substitution and threshold estimates give
\[
 a_{\mathrm{prod}}(n)=B(1+C_U(x))+\Junk(n),
 \qquad0\leq\Junk(n)<B.
\]
Therefore
\[
 \left\lfloor\frac{a_{\mathrm{prod}}(n)}B\right\rfloor-1=C_U(x).
\]
The exceptional branch is correct by the table.

It remains to justify the running time of the postprocessor.  If \(y=0\), the second branch of \(h_\star\) returns zero immediately.  Assume henceforth that \(y\geq1\), let
\(A=4n-4\), and let \(\ell_y=|\binenc(y)|\).
Fix \(s_Q\) with \(2^{s_Q}\leq Q<2^{s_Q+1}\).  If
\(As_Q\geq\ell_y+1\), then \(Q^A>y\), and the quotient is zero.  Otherwise \(A\leq\ell_y\), and exponentiation by squaring computes \(Q^A\) using integers of at most \(A(s_Q+1)\leq2\ell_y\) bits.  Division is polynomial in \(|x|+\ell_y\).  No loop iterates \(n\) in unary.

Thus \((r_{\mathrm{met}},h_\star)\) is a polynomial-time metric reduction.  The function \(C_U\) is parsimoniously \(\#\EXP\)-complete by
Theorem~\ref{thm:source-machine}; composing its parsimonious reductions with this metric reduction proves hardness of \(G_{\mathrm{prod}}\).

For any sufficiently long \(x\), the nonexceptional postprocessor sends \(y=1\) to zero because \(B\geq2\).  Hence positivity of its second argument is not preserved.
\end{proof}

\section{Completeness of the Square-Tiling Count}\label{sec:completion}

The main theorem follows by composing the Wang-level readout with the geometric counting identity. The geometric postprocessor outputs an integer no larger than its input, so the two reductions compose in polynomial bit time. The general membership bound supplies the matching upper bound.

\begin{proof}[Proof of Theorem~\ref{thm:main}]
Every member of \(\Tstar\) is a nonempty finite 4-connected set by Theorem~\ref{thm:geometry}.  Proposition~\ref{prop:fixed-family-sharpE} applies to this fixed family and gives the stronger membership \(G_{\Tstar}\in\#\mathrm{E}\subseteq\#\EXP\).  On the canonical input \(\binenc(0)\), this count includes the unique empty tiling of \(\Square_0\), while malformed strings receive value zero by definition.

For \(n\geq1\), the sets \(\Square_n\) and \(\Rect_{n,n}\) coincide, as do their open Wang matching rules; hence \(a_{\mathrm{prod}}(n)=a_{\mathrm{rect}}(n,n)\).
Theorem~\ref{thm:geometry} gives
\[
 \left\lfloor
 \frac{g_{\Tstar}(P_sn)}
      {s_{\mathrm{amp}}^{2n(n-1)}}\right\rfloor
 =a_{\mathrm{prod}}(n)
\]
for every \(n\geq1\), including \(n=1\) under \(s_{\mathrm{amp}}^0=1\).

Define the instance map by
\[
 r_{\mathrm{geom}}(z)=
 \begin{cases}
  \binenc(P_sn),&z=\binenc(n),\ n\geq1,\\
  \binenc(0),&\text{otherwise}.
 \end{cases}
\]
It checks canonical binary syntax and, on a valid positive input, multiplies by the fixed constant \(P_s\), so it is polynomial-time.  Define \(h_{\mathrm{geom}}(z,y)=0\) when \(z\) is malformed or equals \(\binenc(0)\).  For \(z=\binenc(n)\) with \(n\geq1\), let \(A=2n(n-1)\) and choose the fixed integer \(s_s\) satisfying \(2^{s_s}\leq s_{\mathrm{amp}}<2^{s_s+1}\). If \(A=0\), set \(h_{\mathrm{geom}}(z,y)=y\).  If \(A\geq1\) and \(y=0\), set \(h_{\mathrm{geom}}(z,y)=0\).  For \(A,y\geq1\), let
\(\ell_y=|\binenc(y)|\).  If \(As_s\geq\ell_y+1\), then
\(s_{\mathrm{amp}}^A>y\), so set \(h_{\mathrm{geom}}(z,y)=0\).  Otherwise \(A\leq\ell_y\), and exponentiation by squaring computes the divisor, giving \(h_{\mathrm{geom}}(z,y)=\lfloor y/s_{\mathrm{amp}}^A\rfloor\). All intermediate integers have at most \(A(s_s+1)\leq2\ell_y\) bits.  Computing the divisor and quotient is therefore polynomial in
\(|\binenc(n)|+\ell_y\), without iterating \(n\) in unary.  This proves that
\(h_{\mathrm{geom}}\) is a total polynomial-time map.

For \(z=\binenc(n)\) with \(n\geq1\), the geometric floor identity established above gives
\[
 h_{\mathrm{geom}}\bigl(z,G_{\Tstar}(r_{\mathrm{geom}}(z))\bigr)
 =a_{\mathrm{prod}}(n)=G_{\mathrm{prod}}(z).
\]
On every other \(z\), both sides are zero.  Hence
\((r_{\mathrm{geom}},h_{\mathrm{geom}})\) is a one-query metric reduction.

The postprocessor satisfies \(h_{\mathrm{geom}}(z,y)\leq y\).  To make the composition explicit, combine the reduction \((r_{\mathrm{met}},h_\star)\) of Theorem~\ref{thm:wang-completeness} with the geometric reduction by setting
\[
 r(x)=r_{\mathrm{geom}}(r_{\mathrm{met}}(x)),\qquad
 h(x,y)=h_\star\bigl(x,h_{\mathrm{geom}}(r_{\mathrm{met}}(x),y)\bigr).
\]
Because the intermediate value \(h_{\mathrm{geom}}(r_{\mathrm{met}}(x),y)\) has no more bits than \(y\), the composed postprocessor runs in time polynomial in \(|x|+|\binenc(y)|\).  Hence \(C_U\) metric-reduces to \(G_{\Tstar}\).  Since \(C_U\) is parsimoniously \(\#\EXP\)-complete, \(G_{\Tstar}\) is \(\#\EXP\)-hard under polynomial-time metric reductions.  Together with membership, it is complete.

This completes the proof.
\end{proof}

\section{Conclusion and Discussion}

We prove that counting translation-only tilings of an $n\times n$ square by a fixed finite family of connected polyominoes is $\#\mathrm{EXP}$-complete under polynomial-time metric reductions. Together with the efficient tileability criterion for fixed polyomino families \cite{Pak18}, this establishes a sharp separation between feasibility and exact counting when the entire input is a binary side length. The proof shows that a binary side length can select an accepting-path count within a fixed Wang system and that this numerical information survives an ordinary geometric realization as an exact count, extending classical computation-to-tiling encodings to the fixed-square setting. The decisive step is a global counting guarantee obtained by combining a perimeter bound for open Wang rectangles with a seam bound for nontrivial rectangular dissections. Their combination places the total normalized contribution of all noncanonical configurations below one, so integer division recovers the intended count exactly. A companion Lean~4 formalization of the completeness result is available \cite{LeanRepo}.

Future directions include simplifying the fixed polyomino family and extending numerical isolation to other succinct counting problems. More broadly, combining geometric rigidity with global counting bounds may provide a reusable approach to proving counting hardness under fixed local constraints.

\section*{AI Methodology}
OpenAI's ChatGPT was used to improve the language and clarity of the manuscript. A GPT-based multi-agent system assisted in reviewing the mathematical proofs and organizing the arguments into a Lean 4 formalization, particularly for the constructions and correctness proofs in Sections 3--9.

\end{document}